\documentclass{article}
\usepackage{arxiv}
\usepackage[utf8]{inputenc} 
\usepackage[T1]{fontenc}    
\usepackage{hyperref}       
\usepackage{url}            
\usepackage{booktabs}       
\usepackage{amsfonts}       
\usepackage{nicefrac}       
\usepackage{microtype}      
\usepackage{lipsum}		
\usepackage{graphicx}
\usepackage{natbib}
\usepackage{doi}

\usepackage{amsmath,amssymb,amsfonts,amsthm}
\newtheorem{theorem}{Theorem}

\usepackage[table,dvipsnames]{xcolor}
\usepackage[listings,skins,breakable]{tcolorbox}

\definecolor{tolwhite}{rgb}{0.97,0.97,0.97}
\definecolor{tolblue}{HTML}{84C3E3}
\definecolor{tolred}{HTML}{CC6677}
\definecolor{tolsig}{HTML}{E69F00}

\usepackage{soul}
\newcommand{\hlc}[3][black]{{%
    \colorlet{foo}{#2}%
    \textcolor{#1}{
    \sethlcolor{foo}\hl{#3}}}%
}

\makeatletter
 \def\SOUL@hlpreamble{%
 \setul{}{3.5ex}
 \let\SOUL@stcolor\SOUL@hlcolor
 \SOUL@stpreamble
 }
\makeatother

\usepackage{amsmath}

\usepackage{cleveref}

\usepackage{algpseudocode}
\usepackage{amsfonts}
\usepackage{mathtools,nccmath}
\usepackage{float}

\usepackage{csquotes}
\definecolor{burntorange}{RGB}{200, 90, 0}
\definecolor{gray}{RGB}{150,150,150}
\usepackage{tabularx}
\usepackage{makecell}

\usepackage{algpseudocode}
\usepackage{algcompatible}
\usepackage{algorithm, setspace}
\usepackage{comment}
\usepackage{amssymb}
\usepackage{tikz}

\usepackage{empheq}
\newtcbox{\mymath}[1][]{%
    nobeforeafter, tcbox raise base,
    enhanced, colframe=blue!30!black,
    colback=blue!30, boxrule=1pt,
    #1}
\usepackage{multirow}
\usepackage{booktabs}
\usepackage{amsthm}

\usepackage{soul}

\makeatletter
\newcommand{\algcolor}[2]{%
  \hskip-\ALG@thistlm\colorbox{#1}{\parbox{\linewidth-2\fboxsep}{\hskip\ALG@thistlm\relax #2}}%
}

\makeatother

\usepackage{mathabx}
\usepackage{wasysym}

\title{What are You Weighting For? Improved Weights for Gaussian Mixture Filtering With Application to Cislunar Orbit Determination}

\author{
    {Dalton Durant}\thanks{Ph.D. Student} \\
	Department of Aerospace Engineering \& \\ Engineering Mechanics \\
	The University of Texas at Austin \\
	Austin, TX 78712 \\
	\href{mailto:ddurant@utexas.edu}{\texttt{ddurant@utexas.edu}} \\
	\And
	{Andrey A. Popov}\thanks{Postdoctoral Research Fellow} \\
	The Oden Institute for \\ Computational Engineering \& Sciences \\
	The University of Texas at Austin \\
	Austin, TX 78712 \\
	\href{mailto:andrey.a.popov@utexas.edu}{\texttt{andrey.a.popov@utexas.edu}} \\
    \And
	{Renato Zanetti}\thanks{Associate Professor} \\
	Department of Aerospace Engineering \& \\ Engineering Mechanics \\
	The University of Texas at Austin \\
	Austin, TX 78712 \\
	\href{mailto:renato@utexas.edu}{\texttt{renato@utexas.edu}} \\
}

\renewcommand{\shorttitle}{What are You Weighting For?}

\hypersetup{
pdftitle={What are You Weighting For? Improved Weights for Gaussian Mixture Filtering With Application to Cislunar Orbit Determination},
pdfsubject={},
pdfauthor={Dalton Durant, Andrey A. Popov, Renato Zanetti},
pdfkeywords={Bayesian Estimation, Cislunar Orbit Determination, Gaussian Mixture Models, Nonlinear Filtering},
}

\begin{document}
\maketitle

\begin{abstract}This work focuses on the critical aspect of accurate weight computation during the measurement incorporation phase of Gaussian mixture filters. The proposed novel approach computes weights by linearizing the measurement model about each component's posterior estimate rather than the the prior, as traditionally done. This work proves equivalence with traditional methods for linear models, provides novel sigma-point extensions to the traditional and proposed methods, and empirically demonstrates improved performance in nonlinear cases. Two illustrative examples, the Avocado and a cislunar single target tracking scenario, serve to highlight the advantages of the new weight computation technique by analyzing filter accuracy and consistency through varying the number of Gaussian mixture components.
\end{abstract}


\section{INTRODUCTION}

This work is an adaptation of a previous conference paper \cite{ref:durant2024} and has been expanded to include the following new contributions: bayesian recursive weight updates, sigma point weight updates, pseudocode, and a cislunar single target tracking application.

Gaussian mixture-type filters mark a departure from traditional Kalman filtering methods, which prove more effective for systems characterized by Gaussian dynamics and observations \cite{ref:kalman, ref:julier, ref:van, ref:gelb, ref:lefebvre, ref:arasaratnam2007, ref:arasaratnam2009, ref:zanetti2015}. 
In real-world scenarios, such as weather tracking and orbit determination, nonlinearity and multimodality are common; challenging the viability of Gaussian assumptions \cite{ref:anderson, ref:yun2022, ref:yun2022_2, ref:popov, ref:popov2022, ref:popov2023, ref:durant, ref:li2011, ref:reifler2021, ref:reifler2023}. 
Addressing these challenges, Gaussian mixture-type filters tackle the intricacies of non-Gaussian state estimation by representing probability distributions as a weighted sum of Gaussian components \cite{ref:sorenson, ref:alspach, ref:yun2019}.

Accurate weight assignment is crucial as it determines the contribution of each Gaussian component to the overall distribution. Correctly computed weights not only ensure alignment with the true state's distribution, but also enhance filter consistency, enabling reliable tracking of dynamic changes and accommodation of uncertainties. Conversely, inaccurate weight computation can lead to subpar filtering performance, causing divergence, filter degeneracy, and ultimately, inaccurate and inconsistent state estimates. This work introduces a novel method for computing the weights of Gaussian mixture-type filters by linearizing the measurement model about each component's posterior estimate rather than the prior. This approach enhances accuracy in weight computation and requires only minimal computational overhead, maintaining the overall efficiency of the process.

This work is structured as follows: Section \ref{sec:background} provides basic background information and equations regarding the two different linearization techniques; Section \ref{sec:linear} demonstrates that, for linear measurement models, the presented method is equivalent to traditional methods; Section \ref{sec:nonlinear} gives an explanation for improved performance with nonlinear measurement models; Section \ref{sec:the-weights} introduces alternative forms along with associated algorithms; Section \ref{sec:sigma} presents sigma point extensions along with associated algorithms; Section \ref{sec:nrho} empirically demonstrates improved performance applied to a cislunar single target tracking example; Section \ref{sec:summary} contains a summary of this work and highlights possible future extensions for the proposed approach.

\section{BACKGROUND}
\label{sec:background}

This work computes the weights of Gaussian mixture-type filters by linearizing the measurement model about each component's posterior estimate rather than the prior. This not only improves accuracy, but also ensures minimal additional computational burden. In this section, we explore the distinctions between linearizing the measurement model around the prior and the posterior estimates.

Suppose the measurement model is of the form:
\begin{equation}
\label{eqn:meas}
    \begin{aligned}
        y\ &=\ h(x)\ +\ \eta, \\
    \end{aligned}
\end{equation}
where, $\eta$ is zero mean with covariance matrix $R$. Equation \eqref{eqn:meas} relates the system's true state $x$ to the sensor measurements $y$ and in practical scenarios is often nonlinear. 

Any non-Gaussian prior probability density function (PDF) can be approximated as a weighted sum of $n$ Gaussians with associated prior weights $w_i^-$ arbitraraly well \cite{ref:barshalom2001}:
\begin{equation}
    \begin{aligned}
        p(x)\ 
        &\approx\ \sum_{i=1}^{n}\ w_i^-\ p_i(x).
    \end{aligned}
\end{equation}
\begin{equation}
    \begin{aligned}
        p_i(x)\
        &=\ \mathcal{N}(x;\ \bar{x}_i,\ \bar{P}_{xx}^{(i)}). \\
    \end{aligned}
\end{equation}
The posterior distribution is given by
\begin{align}
    p(x|y)\ &\propto \sum_{i=1}^{n}\ w_i^-\ p(y|x)\ p_i(x)
    \\
    &= \sum_{i=1}^{n}\ \underbrace{w_i^-\ p_i(y)}_{\propto w_i^+}\ p_i(x|y)
\end{align}
where $p(y|x)$ is the measurement likelihood, $p_i(x|y)$ is the \textit{component's} posterior and $p_i(y)$ is the \textit{component's} measurement marginal. For nonlinear measurements, it is generally not possible to compute $p_i(y)$ and $p_i(x|y)$ exactly, and some approximations need to be made. A  notable exemption is the linear measurement case, for which a closed form solution exist, and the posterior is also a Gaussian sum.

In the nonlinear measurement case, $p_i(y)$ and $p_i(x|y)$ are typically approximated as Gaussian, most ofter via linearization around each component mean (EKF) or via statistical linear regression using a set of deterministic sigma points (UKF, CKF). We note that approximations need to be made twice; first to compute $p_i(x|y)$ and then to compute $p_i(y)$. In this work, we argue that once the component's posterior $p_i(x|y)$ is calculated, it is a more accurate representation of the unknown state than the prior, and hence any approximation made to compute $p_i(y)$ should rely on $p_i(x|y)$ rather than $p_i(x)$.

The approximation to the posterior is also a GMM:
\begin{equation}
    \begin{aligned}
        p(x|y)\ 
        &\approx\ \sum_{i=1}^{n}\ w_i^+\ p_i(x|y),
    \end{aligned}
\end{equation}
\begin{equation}
    \begin{aligned}
        p_i(x|y)\
        &=\ \mathcal{N}(x;\ \hat{x}_i,\ \hat{P}_{xx}^{(i)}), \\
    \end{aligned}
\end{equation}
where the updated weights of the $i$-th component are given by:
\begin{equation}
\label{eqn:wplus}
    \begin{aligned}
        w_i^+\ 
        &=\ \frac{w_i^-\ p_i(y)}{p(y)},
    \end{aligned}
\end{equation}
and, 
\begin{equation}
    \begin{aligned}
        p(y)\ 
        &\approx\ \sum_{j=1}^{n}\ w_j^-\ p_j(y),
    \end{aligned}
\end{equation}
\begin{equation}
    \begin{aligned}
        p_i(y)\
        &=\ \mathcal{N}(y;\ h(x_i),\ P_{yy}^{(i)}), \\
    \end{aligned}
\end{equation}
where $x_i$ and $P_{yy}^{(i)}$ are the $i$-th component's current state and measurement innovation covariance, respectively. To compute this covariance, a common technique is to use a first-order Taylor series expansion of the measurement model around the $i$-th component's prior state estimate $\bar{x}_i$:
\begin{equation}
\label{eqn:tse2}
    \begin{aligned}
        h(x_i)\ &\approx\ h(\bar{x}_i)\ +\ \bar{H}_i(x - \bar{x}_i),
    \end{aligned}
\end{equation}
where $\bar{H}_i=H(\bar{x}_i)$ is the Jacobian matrix, which captures the gradient of the measurement model with respect to the state evaluated at the prior estimate. It then follows that

\begin{equation}
\label{eqn:Pyybar}
    \begin{aligned}
        P_{yy}^{(i)}\ 
        &\approx\ \mathbb{E}
        [\bar{\epsilon}_i\ \bar{\epsilon}_i^T] \\
        &=\ \mathbb{E}
        [(y - h(\bar{x}_i))\ (y - h(\bar{x}_i))^T] \\
        &\approx\ \bar{H}_i\bar{P}_{xx}^{(i)}\bar{H}_i^T\ +\ R \\
        &=\ \bar{P}_{yy}^{(i)}.
    \end{aligned}
\end{equation}

The updated weights of the GMM in \eqref{eqn:wplus} are typically computed using this linearization. Linearization around the prior is a common assumption whose implications have not been explored in depth and that can hinder the performance of Gaussian Sum Filters (GSF) \cite{ref:sorenson, ref:alspach, ref:yun2019}, Ensemble Gaussian Mixture Filters (EnGMF) \cite{ref:anderson, ref:silverman, ref:yun2022, ref:popov, ref:popov2022, ref:popov2023, ref:reifler2023, ref:durant}, and other Gaussian Mixture Model (GMM)---type filters \cite{ref:stordal2011, ref:raihan2018-1, ref:raihan2018-2}.

Better weights can be computed by linearizing the measurement model about each component's posterior estimate $\hat{x}_i$ rather than the prior $\bar{x}_i$:

\begin{equation}
\label{eqn:tse3}
    \begin{aligned}
        h(x_i)\ &\approx\ h(\hat{x}_i)\ +\ \hat{H}_i(x - \hat{x}_i),
    \end{aligned}
\end{equation}
where $\hat{H}_i=H(\hat{x}_i)$ is the Jacobian matrix, which now captures the gradient of the measurement model with respect to the state evaluated at the posterior estimate. It then follows that

\begin{equation}
\label{eqn:Pyyhat}
    \begin{aligned}
        P_{yy}^{(i)}\ 
        &\approx\ \mathbb{E}
        [\hat{\epsilon}_i\ \hat{\epsilon}_i^T] \\
        &=\ \mathbb{E}[(y - h(\hat{x}_i))\ (y - h(\hat{x}_i))^T] \\
        &\approx\ \hat{H}_i\hat{P}_{xx}^{(i)} \hat{H}_i^T + R - \hat{H}_iK_iR - (\hat{H}_iK_iR)^T \\
        &=\ \hat{P}_{yy}^{(i)},
    \end{aligned}
\end{equation}

\noindent where $K_i\ =\ \bar{P}_{xx}^{(i)} \bar{H}_i^T\bar{P}_{yy}^{(i)^{-1}}$.

Using the models from \eqref{eqn:tse2} and \eqref{eqn:tse3}, and their associated covariances from \eqref{eqn:Pyybar} and \eqref{eqn:Pyyhat}, two approximations, $\bar{p}(y)$ and $\hat{p}(y)$, are explored in lieu of the elusive true distribution $p(y)$: 

\begin{equation}
\label{eqn:pybar}
    \begin{aligned}
        \bar{p}(y)\ 
        &\approx\ \sum_{j=1}^{n}\ w_j^-\ \bar{p}_j(y),
    \end{aligned}
\end{equation}
\begin{equation}
\label{eqn:pybari}
    \begin{aligned}
        \bar{p}_i(y)\
        &=\ \mathcal{N}(y;\ h(\bar{x}_i),\ \bar{P}_{yy}^{(i)}), \\
    \end{aligned}
\end{equation}
and,
\begin{equation}
\label{eqn:pyhat}
    \begin{aligned}
        \hat{p}(y)\ 
        &\approx\ \sum_{j=1}^{n}\ w_j^-\ \hat{p}_j(y),
    \end{aligned}
\end{equation}
\begin{equation}
\label{eqn:pyhati}
    \begin{aligned}
        \hat{p}_i(y)\
        &=\ \mathcal{N}(y;\ h(\hat{x}_i),\ \hat{P}_{yy}^{(i)}). \\
    \end{aligned}
\end{equation}

This work proposes employing $\hat{p}(y)$ rather than $\bar{p}(y)$ to compute GMM weights and proves that, in linear cases, both approximations yield equivalent results. Additionally, this work empirically demonstrates, in nonlinear cases, improved GMM filter performance when linearizing the measurement model about the posterior estimates. 
The improvement is observed when linearizing the measurement model about the posterior estimates offers a more accurate approximation of the truth compared to linearizing about the prior.

\section{PROOF OF EQUIVALENCE FOR LINEAR MODELS}
\label{sec:linear}

This section proves, for linear measurement models, that using the posterior estimate in the weight update leads to weights that are equivalent to those obtained by using the prior.

\begin{theorem}(Equivalent Weights Under Linear Measurement Models) Consider Gaussian measurement probability distributions $\bar{p}(y)$ and $\bar{p}_i(y)$. These distributions are computed from \eqref{eqn:pybar} and \eqref{eqn:pybari}, respectively, based on prior estimates. Similarly, consider $\hat{p}(y)$ and $\hat{p}_i(y)$, calculated from \eqref{eqn:pyhat} and \eqref{eqn:pyhati}, respectively, utilizing posterior estimates. In light of these considerations, and given the prior weights $w_i^-$, the traditional weights take the form:

\begin{equation}
\label{eqn:wtrad}
    \begin{aligned}
        \bar{w}_i^+\ 
        &=\ \frac{w_i^-\ \bar{p}_i(y)}{\bar{p}(y)},
    \end{aligned}
\end{equation}
and the improved weights have the form:
\begin{equation}
\label{eqn:wstar}
    \begin{aligned}
        \hat{w}_i^+\ 
        &=\ \frac{w_i^-\ \hat{p}_i(y)}{\hat{p}(y)},
    \end{aligned}
\end{equation}
where both $\bar{w}_i^+$ and $\hat{w}_i^+$ share the same prior weights $w_i^-$. The linear case yields $y\ =\ Hx_i\ +\ \eta$, $H(\hat{x}_i)\ =\ H(\bar{x}_i)\ =\ H$ and will give:
\begin{equation}
\label{eqn:w-vs-wstar-LINEAR}
    \begin{aligned}
        \hat{w}_i^+\
        &=\ \frac{w_i^-\ \hat{p}_i(y)}{\hat{p}(y)}\ =\ \frac{w_i^-\ \bar{p}_i(y)}{\bar{p}(y)}\ =\ \bar{w}_i^+. \\
    \end{aligned}
\end{equation}
\end{theorem}

\begin{proof}
Starting from \eqref{eqn:w-vs-wstar-LINEAR}, the prior weights cancel on both sides:
\begin{equation}
    \begin{aligned}
        \frac{\hat{p}_i(y)}{\hat{p}(y)}\ =\ \frac{\bar{p}_i(y)}{\bar{p}(y)}, \\
    \end{aligned}
\end{equation}
then taking the log of both sides yields, 
\begin{align}
\label{eqn:log-ratios}
    \log{\hat{p}_i(y)}\ -\ \log{\hat{p}(y)}\ &=\ \log{\bar{p}_i(y)}\ -\ \log{\bar{p}(y)}.
\end{align}

\noindent Since the measurement Jacobians are not state dependent, this results in constant measurement innovation covariances such that $\hat{P}_{yy}^{(i)} = \hat{P}_{yy}$ and $\bar{P}_{yy}^{(i)} = \bar{P}_{yy}$. Equation \eqref{eqn:log-ratios} simplifies to:
\begin{equation}
    \begin{aligned}
        \frac{1}{2}\ \hat{\epsilon}_i^{T}\ \hat{P}_{yy}^{-1}\ \hat{\epsilon}_i\ +\ \log{\sum_{j=1}^{n} w_j^- \exp{-\frac{1}{2}\ \hat{\epsilon}_j^{T}\ \hat{P}_{yy}^{-1}\ \hat{\epsilon}_j}} \\
        =\ 
        \frac{1}{2}\ \bar{\epsilon}_i^{T}\ \bar{P}_{yy}^{-1}\ \bar{\epsilon}_i\ +\ \log{\sum_{j=1}^{n} w_j^- \exp{-\frac{1}{2}\ \bar{\epsilon}_j^{T}\ \bar{P}_{yy}^{-1}\ \bar{\epsilon}_j}},
    \end{aligned}
\end{equation}
where the innovations $\hat{\epsilon}_i = y - H\hat{x}_i$ and $\bar{\epsilon}_i = y - H\bar{x}_i$. So now, it is sufficient to prove the following to prove \eqref{eqn:w-vs-wstar-LINEAR}:
\begin{equation}
\label{eqn:lin-prove-me}
    \begin{aligned}
        \hat{\epsilon}_i^{T}\ \hat{P}_{yy}^{-1}\ \hat{\epsilon}_i\
        =\ 
        \bar{\epsilon}_i^{T}\ \bar{P}_{yy}^{-1}\ \bar{\epsilon}_i.
    \end{aligned}
\end{equation}

\noindent Reference \cite{ref:zanetti2015} has previously established the proof of \eqref{eqn:lin-prove-me} for linear measurement models. Inherently, this also serves as a natural validation for the proof of \eqref{eqn:w-vs-wstar-LINEAR} in the contributions presented in this work.
\end{proof}

For linear models, we can now conclude that calculating the weights of the posterior distribution from the posterior Gaussian components is equivalent to the traditional method of using the prior Gaussian components, both produce exact Bayesian posteriors for linear, Gaussian Mixture systems.

\section{NONLINEAR MODELS}
\label{sec:nonlinear}

The nonlinear case $y\ =\ h(x_i)\ +\ \eta$ results in $\hat{P}_{yy}^{(i)} \neq \text{const.}$, $\bar{P}_{yy}^{(i)} \neq \text{const.}$, and $H(\hat{x}_i) \neq H(\bar{x}_i)$. This means terms do not simplify as they did for the linear case and we can only conclude that:
\begin{equation}
    \begin{aligned}
        \hat{w}_i^+\
        &=\ \frac{w_i^-\ \hat{p}_i(y)}{\hat{p}(y)}\ \neq\ \frac{w_i^-\ \bar{p}_i(y)}{\bar{p}(y)}\ =\ \bar{w}_i^+. \\
    \end{aligned}
\end{equation}
However, we can demonstrate empirically that:
\begin{equation}
\label{eqn:ibt}
    \begin{aligned}
        \hat{w}_i^+\
        &=\ \frac{w_i^-\ \hat{p}_i(y)}{\hat{p}(y)}\ \ \textit{i.b.t}\ \ \frac{w_i^-\ \bar{p}_i(y)}{\bar{p}(y)}\ =\ \bar{w}_i^+, \\
    \end{aligned}
\end{equation}
where `\textit{i.b.t}' (\textit{is better than}) is used here to mean the new weights $\hat{w}_i^+$ are improved or better than the traditional weights $\bar{w}_i^+$. \footnote{This does not imply that the improved weights are greater than the traditional weights, but only that they have been improved.} This is because the posterior is assumed to be a better approximation of the truth, thus the linearization and by extension the weights are a more accurate representation of the truth. The set of examples in this work provide empirical evidence for improving the weights independently for each GMM component; improving the precision of the overall conditional mean $\mathbb{E}[x|y]$.

\section{COMPUTING THE IMPROVED WEIGHTS}
\label{sec:the-weights}

The improved weights $\hat{w}_i^+$ are:

\begin{empheq}[box={\mymath[colback=white!30,drop lifted shadow, sharp corners]}]{equation}
    \hat{w}_i^+\ \propto\ w_i^-\ \mathcal{N}(y;\ h(\hat{x}_i),\ \hat{P}_{yy}^{(i)}) 
\end{empheq}

And to explicitly compute them:

\begin{equation}
\label{eqn:better-weights}
    \begin{aligned}
        \hat{w}_i^+\ 
        &=\ \frac{w_i^-\ \mathcal{N}(y;\ h(\hat{x}_i),\ \hat{P}_{yy}^{(i)})}{\hat{p}(y)} \\
        &\approx\ \frac{w_i^-\ \mathcal{N}(y;\ h(\hat{x}_i),\ \hat{P}_{yy}^{(i)})}{\sum_{j=1}^{n} w_j^-\ \mathcal{N}(y;\ h(\hat{x}_j),\ \hat{P}_{yy}^{(j)})}.
    \end{aligned}
\end{equation}
The computational expense associated with the above is comparable to computing the weights with the prior. In essence, the computational overhead is minimal, making it practically negligible.

\subsection{Alternate Covariances}
Some alternate equivalent forms of $\hat{P}_{yy}^{(i)}$ for numerical stability:

\begin{align}
        \label{eqn:Pyy-1}
        \hat{P}_{yy}^{(i)}\ &=\ \hat{H}_i\hat{P}_{xx}^{(i)}\hat{H}_i^T + R - \hat{H}_iK_iR - (\hat{H}_iK_iR)^T,  \\
        \label{eqn:Pyy-2}
         &=\ (\hat{H}_i - \bar{H}_i)\hat{P}_{xx}^{(i)}(\hat{H}_i - \bar{H}_i)^T + R \bar{P}_{yy}^{(i)^{-1}} R^T,  \\
        \label{eqn:Pyy-3}
         &=\ (\hat{H}_i - \bar{H}_i)\hat{P}_{xx}^{(i)}(\hat{H}_i - \bar{H}_i)^T\  \\
        &\ \ \ \ \ +\ (I - \bar{H}_i K_i) \bar{P}_{yy}^{(i)} (I - \bar{H}_i K_i)^{T}. \nonumber
\end{align}
Notice that \eqref{eqn:Pyy-1} is not necessarily enforcing semi-positive definiteness, but is enforcing symmetry. Both \eqref{eqn:Pyy-2} and \eqref{eqn:Pyy-3} enforce semi-positive definiteness and symmetry, however, \eqref{eqn:Pyy-2} requires an explicit inversion of $\bar{P}_{yy}$. Therefore, the recommended most numerically stable option is \eqref{eqn:Pyy-3} which this work is coining the ``Joseph Form'' \cite{ref:barshalom2001} of the measurement innovation covariance for the improved weights.

\subsection{Improved Weights in the Update}
In Table \ref{tbl:algo-ekf} is the pseudocode for computing the updates for GMM-type filters using the Extended Kalman Filter (EKF) \cite{ref:gelb}. It provides color coded steps that relate specifically to the filter using traditional weights and the filter using the new weights, GMM(EKF) and GMM(EKF*), respectively. 

\begin{table}[!ht]
    \caption{Pseudocode for the Traditional Weights and the New Weights of GMM-Type Filters with Individual Component EKF Updates; GMM(EKF) and GMM(EKF*)}
    \label{tbl:algo-ekf}
    \centering
    \begin{minipage}{.7\linewidth}
    \hrulefill
    \begin{algorithmic}[H]\\
    \small
    \setstretch{1}
    \State {\centering \hlc[tolblue]{tolwhite}{GMM(EKF)} \hlc[tolwhite]{tolblue}{GMM(EKF*)}  \par}
    \State \textbf{given}\ $\{w_i^-, \bar{x}_i, \bar{P}_{xx}^{(i)}\}_{i=1}^{n}$ $y$, $R$.
    \vskip0.5em
    \State \textbf{step 1.}\ (compute EKF posterior estimates \& weights)
        \State \hskip1em \text{for}\ {$i = 1 \ \text{to} \ n$} \text{do}\
            \State \hskip2em \textbf{step 1.1.}$\boldsymbol{i}$\ (compute individual posterior estimate)
                \State \hskip3em $\bar{H}_i = H(\bar{x}_i)$,
                
                \State \hskip3em $\bar{P}_{yy}^{(i)} = \bar{H}_i \bar{P}_{xx}^{(i)} \bar{H}_i^T + R$, 
                
                \State \hskip3em $K_i = \bar{P}_{xx}^{(i)} \bar{H}_i^T [\bar{P}_{yy}^{(i)}]^{-1}$,  
                \State \hskip3em $\hat{x}_i = \bar{x}_i + K_i [y - h(\bar{x}_i)]$,
                \State \hskip3em $\hat{P}_{xx}^{(i)} = \bar{P}_{xx}^{(i)} - K_i \bar{H}_i \bar{P}_{xx}^{(i)}.$
                \State \hskip3em \hlc[tolwhite]{tolblue}{$\hat{H}_i = H(\hat{x}_i),$}
                \State \hskip3em \hlc[tolwhite]{tolblue}{$\hat{P}_{yy}^{(i)} = (\hat{H}_i - \bar{H}_i)\hat{P}_{xx}^{(i)}(\hat{H}_i - \bar{H}_i)^T$}
                \State \hskip3em \phantom{$\hat{P}_{yy}^{(i)} = $}\hlc[tolwhite]{tolblue}{$+ (I - \bar{H}_i K_i) \bar{P}_{yy}^{(i)} (I - \bar{H}_i K_i)^{T}$.}
            \State \hskip2em \textbf{step 1.2.}$\boldsymbol{i}$\ (compute unnormalized weights)
            
            \State \hskip3em \hlc[tolblue]{tolwhite}{$\bar{w}_i^+ = w_i^- \mathcal{N}(y; h(\bar{x}_i), \bar{P}_{yy}^{(i)}) $.}
            \State \hskip3em \hlc[tolwhite]{tolblue}{$\hat{w}_i^+ = w_i^- \mathcal{N}(y; h(\hat{x}_i), \hat{P}_{yy}^{(i)}) $.}
        \State \hskip1em \text{end for}
    \vskip0.5em
    \State \textbf{step 2.}\ (normalize weights)
        \State \hskip1em \hlc[tolblue]{tolwhite}{$\bar{w}_i^{+} \coloneqq \bar{w}_i^{+}/(\sum_{j=1}^{n} \bar{w}_j^{+})$, for $i = 1,\cdots,n$.}
        \State \hskip1em \hlc[tolwhite]{tolblue}{$\hat{w}_i^{+} \coloneqq \hat{w}_i^{+}/(\sum_{j=1}^{n} \hat{w}_j^{+})$, for $i = 1,\cdots,n$.}
    \vskip0.5em
    \State \textbf{output}\ \hlc[tolblue]{tolwhite}{$\{ \bar{w}_i^+, \hat{x}_i, \hat{P}_{xx}^{(i)}\}_{i=1}^{n}$}  \hlc[tolwhite]{tolblue}{$\{ \hat{w}_i^+, \hat{x}_i, \hat{P}_{xx}^{(i)}\}_{i=1}^{n}$}.
    \end{algorithmic}
    \hrulefill
    \end{minipage}
    
\end{table}

Similarly, this new weights formulation can be expanded to filters other than the Gaussian sum filter. One recently proposed filter, the Bayesian Recursive Update Filter (BRUF), recursively updates the mean and covariance $N$ times using partitioned measurement information \cite{ref:zanetti2015, ref:michaelson2023-1, ref:michaelson2023-2, ref:michaelson2023-3}. Table \ref{tbl:algo-bruf} provides color coded steps for the GMM(BRUF) using traditional weights and the GMM(BRUF*) using the new weights.

\begin{table}[!ht]
    \caption{Pseudocode for the Traditional Weights and the New Weights of GMM-Type Filters with Individual Component BRUF Updates; GMM(BRUF) and GMM(BRUF*)}
    \label{tbl:algo-bruf}
    \centering
    \begin{minipage}{.7\linewidth}
    \hrulefill
    \begin{algorithmic}[H]\\
    \small
    \setstretch{1}
    \State {\centering \hlc[tolred]{tolwhite}{GMM(BRUF)} \hlc[tolwhite]{tolred}{GMM(BRUF*)} \par}
    \State \textbf{given}\ $\{w_i^-, \bar{x}_i, \bar{P}_{xx}^{(i)}\}_{i=1}^{n}$ $y$, $R$, $N$.
    \vskip0.5em
    \State \textbf{step 1.}\ (compute BRUF posterior estimates \& weights)
        \State \hskip1em \text{for}\ {$i = 1 \ \text{to} \ n$} \text{do}\
            \State \hskip2em \textbf{step 1.1.}$\boldsymbol{i}$\ (compute individual posterior estimate)
                \State \hskip3em $\bar{x}_i^{0} = \bar{x}_i$.
                \State \hskip3em $\bar{P}_0^{(i)} = \bar{P}_{xx}^{(i)}$.
                \State \hskip3em \text{for}\ {$j = 1 \ \text{to} \ N$} \text{do}\
                    \State \hskip4em $\tilde{H}_j = H(\bar{x}_i)$,
                    \State \hskip4em $\tilde{P}_{yy}^{(j)} = \tilde{H}_j \bar{P}_{xx}^{(i)} [\tilde{H}_j]^T + NR$, 
                    \State \hskip4em $\tilde{K}_j = \bar{P}_{xx}^{(i)} [\tilde{H}_j]^T [\tilde{P}_{yy}^{(j)}]^{-1}$,  
                    \State \hskip4em $\bar{x}_i \leftarrow \bar{x}_i + \tilde{K}_j  [y - h(\bar{x}_i)]$,
                    \State \hskip4em $\bar{P}_{xx}^{(i)} \leftarrow \bar{P}_{xx}^{(i)} - \tilde{K}_j \tilde{H}_j \bar{P}_{xx}^{(i)},$
                \State \hskip3em \text{end}
                \State \hskip3em $\hat{x}_i = \bar{x}_i$.
                \State \hskip3em $\hat{P}_{xx}^{(i)} = \bar{P}_{xx}^{(i)}.$

            \State \hskip2em \textbf{step 1.2.}$\boldsymbol{i}$\ (compute measurement innovation covariance)
                \State \hskip3em $\bar{H}_i = H(\bar{x}_i^{0})$,
                \State \hskip3em $\bar{P}_{yy}^{(i)} = \bar{H}_i \bar{P}_0^{(i)} [\bar{H}_i]^T + R$.
                \State \hskip3em \hlc[tolwhite]{tolred}{$K_i = \bar{P}_0^{(i)} [\bar{H}_i]^T [\bar{P}_{yy}^{(i)}]^{-1}$,}
                \State \hskip3em \hlc[tolwhite]{tolred}{$\hat{H}_i = H(\hat{x}_i)$,}
                \State \hskip3em \hlc[tolwhite]{tolred}{$\hat{P}_{yy}^{(i)} = (\hat{H}_i - \bar{H}_i)\hat{P}_{xx}^{(i)}(\hat{H}_i - \bar{H}_i)^T$}
                \State \hskip3em \phantom{$\hat{P}_{yy}^{(i)} = $}\hlc[tolwhite]{tolred}{$+ (I - \bar{H}_i K_i) \bar{P}_{yy}^{(i)} (I - \bar{H}_i K_i)^{T}$.}
                
            \State \hskip2em \textbf{step 1.3.}$\boldsymbol{i}$\ (compute unnormalized weights)
                \State \hskip3em \hlc[tolred]{tolwhite}{$\bar{w}_i^+ = w_i^- \mathcal{N}(y; h(\bar{x}_i^0), \bar{P}_{yy}^{(i)}) $.}
                \State \hskip3em \hlc[tolwhite]{tolred}{$\hat{w}_i^+ = w_i^- \mathcal{N}(y; h(\hat{x}_i), \hat{P}_{yy}^{(i)}) $.}

        \State \hskip1em \text{end}
    \vskip0.5em
    \State \textbf{step 2.}\ (normalize weights)
        \State \hskip1em \hlc[tolred]{tolwhite}{$\bar{w}_i^{+} \coloneqq \bar{w}_i^{+}/(\sum_{j=1}^{n} \bar{w}_j^{+})$, for $i = 1,\cdots,n$.}
        \State \hskip1em \hlc[tolwhite]{tolred}{$\hat{w}_i^{+} \coloneqq \hat{w}_i^{+}/(\sum_{j=1}^{n} \hat{w}_j^{+})$, for $i = 1,\cdots,n$.}
    \vskip0.5em
    \State \textbf{output}\ \hlc[tolred]{tolwhite}{$\{ \bar{w}_i^+, \hat{x}_i, \hat{P}_{xx}^{(i)}\}_{i=1}^{n}$}  \hlc[tolwhite]{tolred}{$\{ \hat{w}_i^+, \hat{x}_i, \hat{P}_{xx}^{(i)}\}_{i=1}^{n}$}.
    \end{algorithmic}
    \hrulefill
    \end{minipage}
\end{table}

\subsection{Simulation Results}

The following example will demonstrate empirically that the new GMM weights proposed can lead to improved performance compared to using the traditional weights. The compared filters will execute a single measurement update in two dimensions, shown by Fig.~\ref{fig:avocado}. This approach aims to demonstrate the effectiveness of the new weight update scheme, showcasing its broad applicability across all GMM-type filters and decoupling it from the effects of resampling, pruning, propagation, and so forth. As such, the results presented are relevant to GSF, EnGMF, and other GMM-type filters of a similar nature.

\begin{figure}[!ht]
    \centering
    \includegraphics[clip, trim=1cm 0cm 1cm 0cm, width=0.6\linewidth]{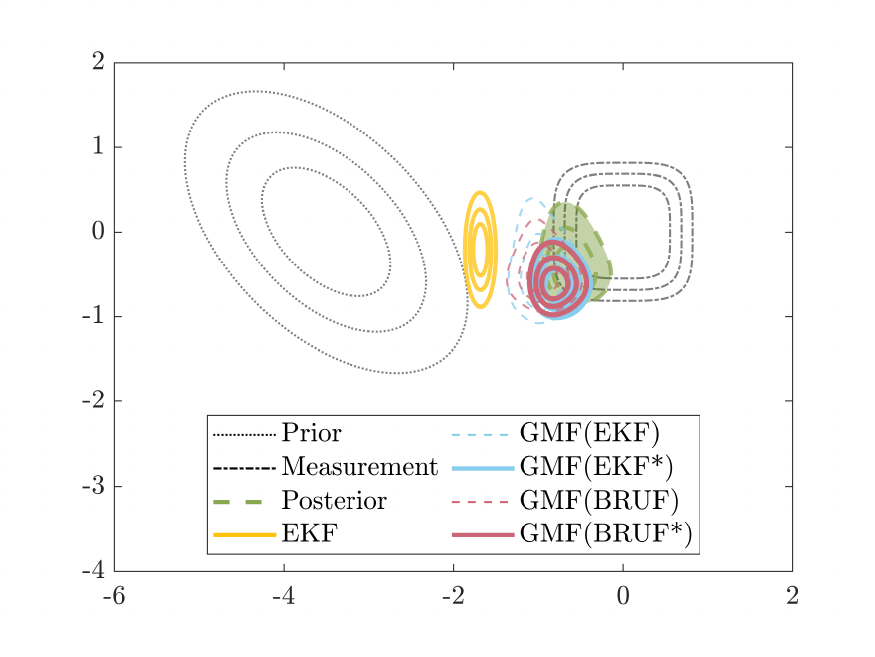}
    \caption{A plot of the Avocado, single update example. Comparing the posterior PDF estimates of the EKF, GMF(EKF), GMF(BRUF), GMF(EKF*), and GMF(BRUF*) averaged over 100 Monte Carlo simulations. The GMM-type filters use 100 components. }
    \label{fig:avocado}
\end{figure}

The prior is expressed as a Gaussuan distribution:
\begin{equation}
    \begin{aligned}
        \mathcal{N} \left ( \begin{bmatrix}-3.5 \\ 0  \end{bmatrix}, \begin{bmatrix} 1 & -\frac{1}{2} \\ -\frac{1}{2} & 1 \end{bmatrix} \right ),
    \end{aligned}
\end{equation}
which is off-center from the origin, bivariate with correlated variables. The nonlinear measurement mapping is of the form
\begin{equation}
    \begin{aligned}
        h(x)\ &=\ 
        \begin{bmatrix}
        (x_1)^2 \\
        (x_2)^2
        \end{bmatrix},
    \end{aligned}
\end{equation}
with a measured value of
\begin{equation}
    \begin{aligned}
        y\ &=\ 
        \begin{bmatrix}
        0 \\
        0
        \end{bmatrix},
    \end{aligned}
\end{equation}
which implies the measured value is not very likely given the prior. Additionally, there is high confidence in the measurement which is assumed corrupted by additive Gaussian noise with zero mean and covariance matrix
\begin{equation}
    \begin{aligned}
        R\ &=\ (0.4^2)I_{2 \times 2},
    \end{aligned}
\end{equation}
where $I_{2 \times 2}$ is the two dimension identity matrix. This example compares the GMF with the new improved weights against the traditional weights. For fair comparison, each GMF has the same number of components, $N$, with uniform prior weights $w_i^- = 1/N$. For notation, the GMF(EKF) is the GSF (or EnGMF, or any GMM-type filter) performing individual EKF updates for each component and is using the traditional weight update described by \eqref{eqn:wtrad}. Conversely, the GMF(EKF*) is the GMM-type filter performing individual EKF updates for each component, but is using the new improved weight update described by \eqref{eqn:wstar} and more explicitly by \eqref{eqn:better-weights}. In similar fashion, we include the GMF(BRUF) and GMF(BRUF*) using the traditional and improved weights, respectively. Lastly, we include the EKF for good measure. The GMF filters (with a sufficient number of components) will outlast a typical EKF for non-Gaussian distributions. 

To assess filter performance, the root mean square error (RMSE) and the Kullback–Leibler divergence (KLD) or relative entropy are used. RMSE is computed by
\begin{equation}
\label{eqn:rmse}
\text{RMSE}\ =\ \sqrt{\frac{1}{n_x} (x - \hat{x})^{T}(x - \hat{x})},
\end{equation}
where $n_x$ is the size of the state-space, $x$ is the truth, and $\hat{x}$ is the weighted sum of each component's posterior mean. KLD is computed by
\begin{equation}
    D_{\text{KL}}(P||Q)\ =\ \frac{1}{s_x}\sum_{x} \frac{1}{2}(\log{P(x|y)}\ -\ \log{Q(x|y)})^2.
\end{equation}
We evaluate the KLD by discretizing it on a $s_x \times s_x$ grid; where $s_x$ is the number of discrete values of $x$ on the grid, and $P(x|y)$ and $Q(x|y)$ are the approximated and true posterior PDF's, respectively. In this example, $P(x|y)$ is the estimated PDF produced by the filter after the update represented by a contour PDF in Fig.~\ref{fig:avocado} and $Q(x|y)$ is the truth represented by the green shaded PDF.

RMSE is a measure of how accurate the state estimate is with respect to the truth. A lower RMSE indicates a more accurate filter. In the context of Table~\ref{tbl:avocado}, using the improved weights, GMF(EKF*) and GMF(BRUF*), both give better accuracy when compared to the EKF and their traditional weight counterparts, GMF(EKF) and GMF(BRUF); suggesting improved performance. 

KLD, on-the-other-hand, is a measure of the dissimilarity between two probability distributions. Though it is not a true distance metric as it is not symmetric and does not satisfy the triangle inequality, it is useful in quantifying how the true probability distribution differs from the approximating distribution. A lower KLD indicates there is less information needed to match the approximated distribution to the truth. Table~\ref{tbl:avocado} shows that using the improved weights can substantially decrease KLD; again improving filter performance. 
\vspace{1em}
\begin{table}[!ht]
\centering
\caption{Numerical output of the single update, two dimensional Avocado example. Comparing the root mean square error and the Kullback–Leibler divergence for the EKF, GMF(EKF), GMF(BRUF), GMF(EKF*), and GMF(BRUF*). Averaged over 100 Monte Carlo simulations. The GMM-type filters use 100 components. }
\normalsize
\begin{tabular}{ccc}
\toprule
            & RMSE  & KLD         \\ \midrule
EKF         & $0.8929$   & ---      \\
GMF(EKF)  & $0.2899$   & $12.594$   \\
GMF(BRUF)  & $0.2874$   & $6.2383$   \\
GMF(EKF*) & $\mathbf{0.2378}$   & $\mathbf{0.8226}$   \\ 
GMF(BRUF*) & $\mathbf{0.2468}$   & $\mathbf{0.6326}$ \\ \bottomrule
\end{tabular}
\label{tbl:avocado}
\end{table}

To aid in our justification, we varied the number of GMM components in the GMF. The RMSE and KLD results vs number of components are shown in Fig.~\ref{fig:avocado-rmse} and Fig.~\ref{fig:avocado-kld}, respectively. Noticeably, the EKF suffers and is not visible in the plotting frames due to it's approximation of Gaussian distributions. This was understood previously from it's poor performance in Fig.~\ref{fig:avocado} and Table~\ref{tbl:avocado}. For the GMM-type filters, with smaller number of components, there is a clear indication of improved performance when using the new weights over the traditional ones. As the number of components increases, a noteworthy trend emerges: the filters utilizing traditional weights, namely GMF(EKF) and GMF(BRUF), appear to gradually converge towards the performance exhibited by the filter employing the improved weights, GMF(EKF*) and GMF(BRUF*). As the number of Gaussian components grows, then the prior covariance of each component becomes smaller, the update of each Gaussian component becomes smaller, and hence the difference between prior and posterior Gaussian components also becomes smaller. In the limit as the number of components goes to infinity, the two should be identical.

\begin{figure}[!ht]
    \centering
    \includegraphics[clip, trim=0cm 0cm 0cm 0cm, width=0.45\linewidth]{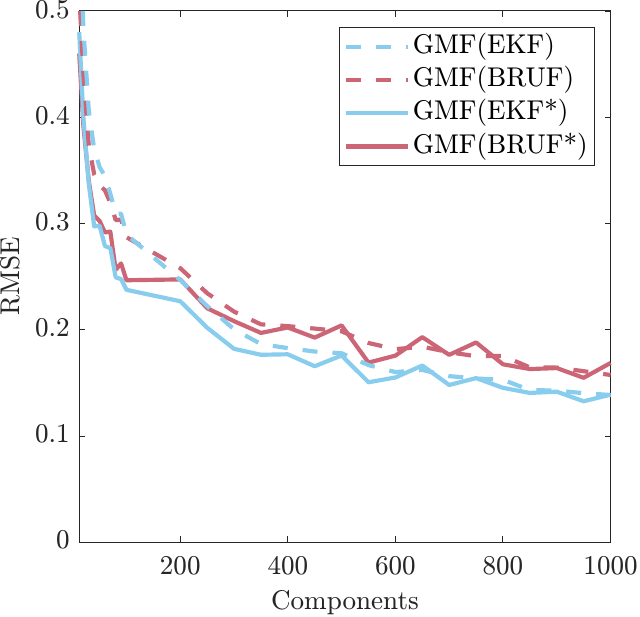}
    \caption{Comparing the impact on state estimate accuracy with-respect-to the truth, represented by root mean square error (RMSE), by varying the number of GMM components in the single update Avocado example. Featuring the EKF, GMF(EKF), GMF(BRUF), GMF(EKF*), and GMF(BRUF*). The EKF results are plotted out of frame and are disregarded. Plots are averaged across 100 Monte Carlo simulations for each GMM component test case. }
    \label{fig:avocado-rmse}
\end{figure}
\begin{figure}[!ht]
    \centering
    \includegraphics[clip, trim=0cm 0cm 0cm 0cm, width=0.45\linewidth]{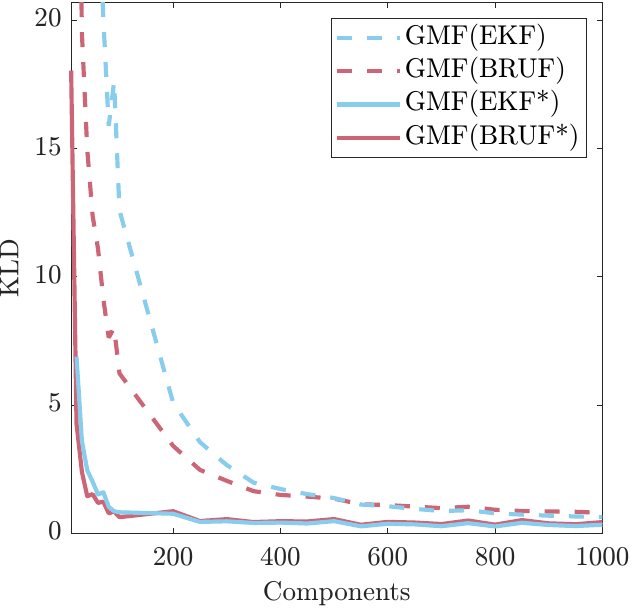}
    \caption{Comparing the impact on PDF accuracy with-respect-to the truth, represented by Kullback–Leibler divergence (KLD), by varying the number of GMM components in the single update Avocado example. Featuring the EKF, GMF(EKF), GMF(BRUF), GMF(EKF*), and GMF(BRUF*). The EKF results are plotted out of frame and are disregarded. Plots are averaged across 100 Monte Carlo simulations for each GMM component test case. }
    \label{fig:avocado-kld}
\end{figure}

\section{EXTENSION TO SIGMA POINT UPDATES}
\label{sec:sigma}
The conventional formulations of traditional and improved weights, as currently presented, exhibit limited applicability, particularly to filters that do not calculate Jacobians, such as sigma point filters. However, an innovative breakthrough arises as effective techniques from importance sampling are incorporated into the methodology. Herein lies a novel approach to performing the GMM weight update using sigma points, encompassing both traditional and improved weight methods. 

Using the unscented transform \cite{ref:julier, ref:van} to generate sigma points $\bar{\chi}_{i_l}$, for $l=0,...,2n_x$ around the prior $p_i(x)$, with mean $\bar{x}_i$ and covariance $\bar{P}_{xx}^{(i)}$, the traditional weights can be formulated by \cite{ref:kottakki2014}
\begin{equation}
    \begin{aligned}
        \bar{w}_i^+ &\approx \frac{w_i^-}{\bar{p}(y)} \mathcal{N}(y;\sum_{l=0}^{2n_x} W_{m_l}h(\bar{\chi}_{i_l}),\bar{P}_{yy}^{(i)}),
    \end{aligned}
\end{equation}
or by 
\begin{equation}
    \label{eqn:w-sigma}
    \begin{aligned}
        \bar{w}_i^+ &\approx \frac{w_i^-}{\bar{p}(y)} \sum_{l=0}^{2n_x} W_{m_l} \mathcal{N}(y;h(\bar{\chi}_{i_l}),\bar{P}_{yy}^{(i)}),
    \end{aligned}
\end{equation}
which both forms represent the weighted contributions of each sigma point to the $ith$ component's measurement distribution.

In this work, we propose alternative approximations that perform better under the numerical examples studied. Using techniques from importance sampling, we can reformulate and make the weight update better. 
Starting from \eqref{eqn:wplus}
\begin{equation}
    \begin{aligned}   
        w_i^+ &= \frac{w_i^- p_i(y)}{p(y)} \\
        &\propto \int_{\mathbb{R}^{n_x}} p_i(x,y) dx \\
        &= \int_{\mathbb{R}^{n_x}} p_i(x) p_i(y|x)dx \\
        &= \mathbb{E}_{x \sim p_i(x)} \left [ p_i(y|x)   \right ].
    \end{aligned}
\end{equation}
By using the unscented transform to generate sigma points $\bar{\chi}_{i_l}$ for $l=0,..., 2n_x$ around the prior $p_i(x)$, the traditional weights become
\begin{equation}
    \label{eqn:w_bar_sigma}
    \begin{aligned}           
        \bar{w}_i^+\ &\approx \frac{w_i^-}{\bar{p}(y)} \sum_{l=0}^{2n_x} W_{m_l} p_i(y|x=\bar{\chi}_{i_l}).
    \end{aligned}
\end{equation}
This approach avoids the need to compute $\bar{P}_{yy}^{(i)}$ for the weights and directly passes the sigma points into the likelihood $p_i(y|x)$.

We can also use this same technique for the new proposed weights.
\begin{equation}
    \begin{aligned}   
        w_i^+ &= \frac{w_i^- p_i(y)}{p(y)} \\
        &\propto \int_{\mathbb{R}^{n_x}} p_i(x,y) dx \\
        &= \int_{\mathbb{R}^{n_x}} p_i(x) p_i(y|x)dx \\
\text{but instead,} \\
        &= \int_{\mathbb{R}^{n_x}} p_i(x) p_i(y|x) \frac{p_i(x|y)}{p_i(x|y)} dx \\
        &= \mathbb{E}_{x \sim p_i(x|y)} \left [ \frac{p_i(x)}{p_i(x|y)} p_i(y|x)   \right ].
    \end{aligned}
\end{equation}
Now, using the unscented transform to generate sigma points $\hat{\chi}_{i_l}$ for $l=0,..., 2n_x$ around the posterior $p_i(x|y)$, with mean $\hat{x}_i$ and covariance $\hat{P}_{xx}^{(i)}$, the new weights become

\begin{equation}
    \label{eqn:w-star-sigma}
    \begin{aligned}           
        \hat{w}_i^+ &\approx \frac{w_i^-}{\hat{p}(y)} \sum_{l=0}^{2n_x} W_{m_l} \frac{p_i(x=\hat{\chi}_{i_l}) p_i(y|x=\hat{\chi}_{i_l})}{p_i(x=\hat{\chi}_{i_l}|y)}.
    \end{aligned}
\end{equation}
Similar to \eqref{eqn:w_bar_sigma}, this avoids the need to compute $\hat{P}_{yy}^{(i)}$ to calculate the weights. 

Table \ref{tbl:algo-sigma} contains the pseudocode for computing the updates for GMM-type filters using sigma point filters. This algorithm can be used for any sigma point filter provided it is in a three parameter formulation \cite{ref:van}. Table \ref{tbl:algo-sigma} provides color coded steps for the GMM($\sigma$) using traditional weights from \eqref{eqn:w-sigma}, and for the GMM($\sigma$*) using the new weights from \eqref{eqn:w-star-sigma}.
\vspace{1em}
\begin{table}[!ht]
    \caption{Pseudocode for the Traditional Weights and the New Weights of GMM-Type Filters with Individual Component Sigma Point Filter Updates; GMM($\sigma$) and GMM($\sigma$*)}
    \label{tbl:algo-sigma}
    \centering
    \begin{minipage}{.7\linewidth}
    \hrulefill
    \begin{algorithmic}[H]\\
    \small
    \setstretch{1}
    \State {\centering \hlc[tolsig]{tolwhite}{GMM($\sigma$)} \hlc[tolwhite]{tolsig}{GMM($\sigma$*)} \par}
    \State \textbf{given}\ $\{w_i^-, \bar{x}_i, \bar{P}_{xx}^{(i)}\}_{i=1}^{n}$ $y$, $R$, $\alpha$, $\beta$, $\kappa$, $n_x$.
    \vskip0.5em
    \State \textbf{step 1.}\ (compute sigma point filter posterior estimates \& weights)
        \State \hskip1em $\lambda_u = \alpha^2 (n_x + \kappa) - n_x$
        \State \hskip1em $W_{m_0} = \lambda_u / (n_x + \lambda_u)$
        \State \hskip1em $W_{c_0} = \lambda_u / (n_x + \lambda_u) + (1 - \alpha^2 + \beta)$
        \State \hskip1em $W_{c_i} = W_{m_i} = 1/(2(n_x+\lambda_u))$, for $i = 1,\cdots,2n_x$.
        \State \hskip1em \text{for}\ {$i = 1 \ \text{to} \ n$} \text{do}\
            \State \hskip2em \textbf{step 1.1.}$\boldsymbol{i}$\ (compute individual posterior estimate)
                \State \hskip3em $\bar{\chi}_i =  [ \bar{x}_i \ \bar{x}_i \pm ((n_x+\lambda_u)\bar{P}_{xx}^{(i)})^{1/2}  ]$
                \State \hskip3em $\bar{y}_i = \sum_{l=0}^{2n_x} W_{m_l} h(\bar{\chi}_{i_l})$,
                \State \hskip3em $\bar{P}_{yy}^{(i)} = \sum_{l=0}^{2n_x} W_{c_l} [h(\bar{\chi}_{i_l}) - \bar{y}_i][h(\bar{\chi}_{i_l}) - \bar{y}_i]^T + R$, 
                \State \hskip3em $\bar{P}_{xy}^{(i)} = \sum_{l=0}^{2n_x} W_{c_l} [\bar{\chi}_{i_l} - \bar{x}_i][h(\bar{\chi}_{i_l}) - \bar{y}_i]^T$,
                \State \hskip3em $K_i = \bar{P}_{xy}^{(i)} 
                [\bar{P}_{yy}^{(i)}]^{-1}$, 
                \State \hskip3em $\hat{x}_i = \bar{x}_i + K_i [y - \bar{y}_i]$,
                \State \hskip3em $\hat{P}_{xx}^{(i)} = \bar{P}_{xx}^{(i)} - K_i \bar{H}_i \bar{P}_{xx}^{(i)}.$
                \State \hskip3em \hlc[tolwhite]{tolsig}{$\hat{\chi}_i = [ \hat{x}_i \ \hat{x}_i \pm ((n_x+\lambda_u)\hat{P}_{xx}^{(i)})^{1/2} ]$.}

            \State \hskip2em \textbf{step 1.2.}$\boldsymbol{i}$\ (compute unnormalized weights)
                \State \hskip3em \hlc[tolsig]{tolwhite}{$\bar{w}_i^+ = w_i^- \sum_{l=0}^{2n_x} W_{m_l} \mathcal{N}  (y; h(\bar{\chi}_{i_l}), \bar{P}_{yy}^{(i)}  ) $.}
                \State \hskip3em \hlc[tolwhite]{tolsig}{$\hat{w}_i^+ = w_i^- \sum_{l=0}^{2n_x} W_{m_l} \mathcal{N} \left (\hat{\chi}_{i_l}; \bar{x}_i, \bar{P}_{xx}^{(i)} \right )$}
                \State \hskip3em \phantom{$\hat{w}_i^+ = $}\hlc[tolwhite]{tolsig}{$\circ \ \mathcal{N} \left (y; h(\hat{\chi}_{i_l}), R \right ) / \mathcal{N} \left (\hat{\chi}_{i_l}; \hat{x}_i, \hat{P}_{xx}^{(i)} \right )$.}

        \State \hskip1em \text{end}
    \vskip0.5em
    \State \textbf{step 2.}\ (normalize weights)
        \State \hskip1em \hlc[tolsig]{tolwhite}{$\bar{w}_i^{+} \coloneqq \bar{w}_i^{+}/(\sum_{j=1}^{n} \bar{w}_j^{+})$, for $i = 1,\cdots,n$.}
        \State \hskip1em \hlc[tolwhite]{tolsig}{$\hat{w}_i^{+} \coloneqq \hat{w}_i^{+}/(\sum_{j=1}^{n} \hat{w}_j^{+})$, for $i = 1,\cdots,n$.}
    \vskip0.5em
    \State \textbf{output}\ \hlc[tolsig]{tolwhite}{$\{ \bar{w}_i^+, \hat{x}_i, \hat{P}_{xx}^{(i)}\}_{i=1}^{n}$}  \hlc[tolwhite]{tolsig}{$\{ \hat{w}_i^+, \hat{x}_i, \hat{P}_{xx}^{(i)}\}_{i=1}^{n}$}.
    \end{algorithmic}
    \hrulefill
    \end{minipage}
\end{table}

\subsection{Simulation Results for Sigma Point Updates}

A single update is performed for the Avocado example, shown in Fig.~\ref{fig:avocado-sigma}, and, similar to before, the results presented are relevant for GSF, EnGMF, and other GMM-type filters of a similar nature. The difference with the prior example is that now the individual Gaussian components are updated using the Unscented Kalman Filter (UKF) and the Cubature Kalman Filter (CKF); two types of sigma point filters. We include the UKF for good measure with tuning parameters $\alpha = 1$, $\beta = 2$, and $\kappa = 3$ \cite{ref:van}. The GMF(UKF) uses the same tuning parameters as the UKF and is using the traditional weight update. Also included is the GMF(CKF), whose components can be expressed as UKFs with parameters $\alpha = 1$, $\beta = 0$, and $\kappa = 0$. Additionally, we showcase these filters with improved weight versions, GMF(UKF*) and GMF(CKF*). The GMF filters (with a sufficient number of components) will outlast a typical UKF for non-Gaussian distributions. 

\begin{figure}[!ht]
    \centering
    \hspace{-1em}
    \includegraphics[clip, trim=1cm 0cm 1cm 0cm, width=0.6\linewidth]{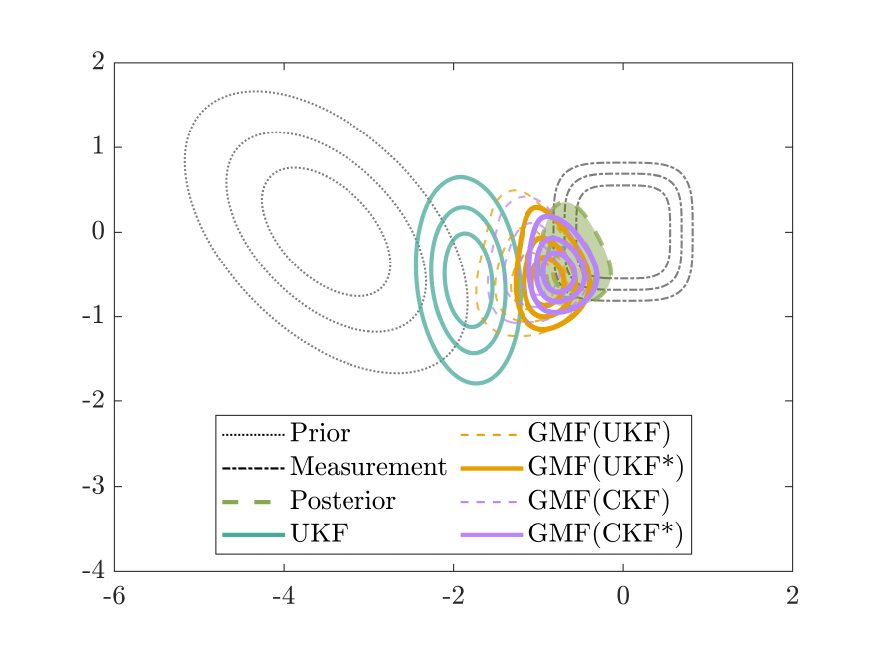}
    \caption{A plot of the Avocado, single update example. Comparing the posterior PDF estimates of the UKF, GMF(UKF), GMF(CKF), GMF(UKF*), and GMF(CKF*) averaged over 100 Monte Carlo simulations. The GMM-type filters use 100 components. }
    \label{fig:avocado-sigma}
\end{figure}

In Table~\ref{tbl:avocado-sigma}, using the improved weights, GMF(UKF*) and GMF(CKF*), both give better accuracy when compared to the UKF and their traditional weight counterparts, GMF(UKF) and GMF(CKF); suggesting improved performance. Additionally, Table~\ref{tbl:avocado-sigma} shows that using the improved weights can substantially decrease KLD; improving filter performance. 
\vspace{1em}
\begin{table}[!ht]
\centering
\caption{Numerical output of the single update, two dimensional Avocado example. Comparing the root mean square error and the Kullback–Leibler divergence for the UKF, GMF(UKF), GMF(CKF), GMF(UKF*), and GMF(CKF*). Averaged over 100 Monte Carlo simulations. The GMM-type filters use 100 components. }
\normalsize
\begin{tabular}{ccc}
\toprule
            & RMSE  & KLD         \\ \midrule
UKF         & $1.0138$   & ---      \\
GMF(UKF)  & $0.4322$   & $93.429$   \\
GMF(CKF)  & $0.3592$   & $46.429$   \\
GMF(UKF*) & $\mathbf{0.2679}$   & $\mathbf{4.0853}$   \\ 
GMF(CKF*) & $\mathbf{0.1770}$   & $\mathbf{0.8941}$ \\ \bottomrule
\end{tabular}
\label{tbl:avocado-sigma}
\end{table}

As done for the prior example, we varied the number of GMM components in the GMF. The RMSE and KLD results vs number of components are shown in Fig.~\ref{fig:avocado-rmse-sigma} and Fig.~\ref{fig:avocado-kld-sigma}, respectively. Noticeably, the UKF suffers and is not visible in the plotting frames due to it's approximation of Gaussian distributions. This poor performance is also reflected in Fig.~\ref{fig:avocado-sigma} and Table~\ref{tbl:avocado-sigma}. For the GMM-type filters, with smaller number of components, there is a clear indication of improved performance when using the new improved weights over the traditional ones. However, for the sigma point filters, as the number of components increases, the filters utilizing traditional weights, namely GMF(UKF) and GMF(CKF), do not converge as rapidly towards the performance exhibited by the filter employing the improved weights, GMF(UKF*) and GMF(CKF*).

\begin{figure}[!ht]
    \centering
    \includegraphics[clip, trim=0cm 0cm 0cm 0cm, width=0.45\linewidth]{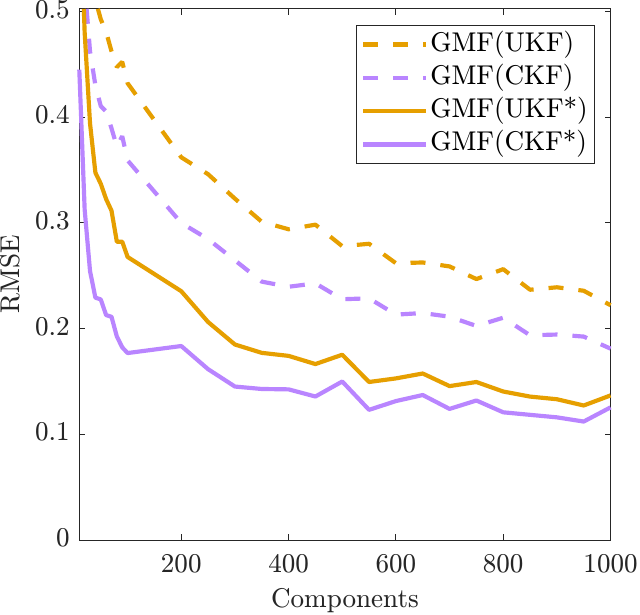}
    \caption{Comparing the impact on state estimate accuracy with-respect-to the truth, represented by root mean square error (RMSE), by varying the number of GMM components in the single update Avocado example. Featuring the UKF, GMF(UKF), GMF(CKF), GMF(UKF*), and GMF(CKF*). The UKF results are plotted out of frame and are disregarded. Plots are averaged across 100 Monte Carlo simulations for each GMM component test case. }
    \label{fig:avocado-rmse-sigma}
\end{figure}
\begin{figure}[!ht]
    \centering
    \includegraphics[clip, trim=0cm 0cm 0cm 0cm, width=0.45\linewidth]{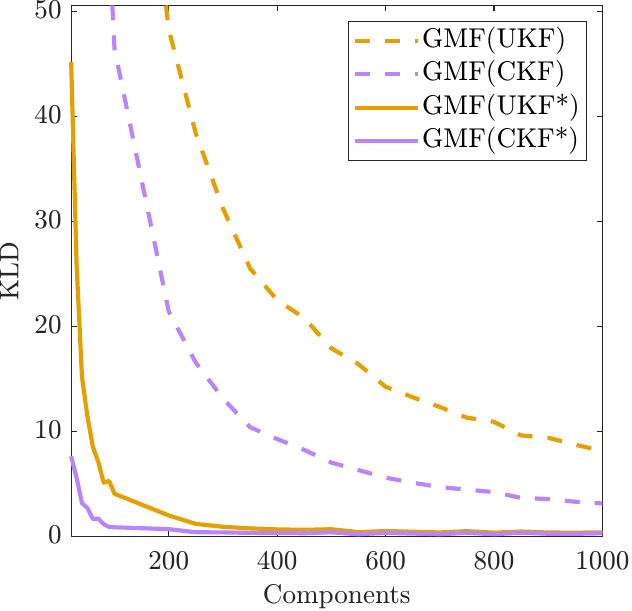}
    \caption{Comparing the impact on PDF accuracy with-respect-to the truth, represented by Kullback–Leibler divergence (KLD), by varying the number of GMM components in the single update Avocado example. Featuring the UKF, GMF(UKF), GMF(CKF), GMF(UKF*), and GMF(CKF*). The UKF results are plotted out of frame and are disregarded. Plots are averaged across 100 Monte Carlo simulations for each GMM component test case. }
    \label{fig:avocado-kld-sigma}
\end{figure}

\section{A CISLUNAR NRHO EXAMPLE}
\label{sec:nrho}
A cislunar Near-Rectilinear Halo Orbit (NRHO) example will now demonstrate the strength of the improved GMM weighting scheme in an aerospace application. It focuses on comparing filters in the complex cislunar NRHO system, known for its chaotic behavior and sensitivity to initial conditions. Minor errors in the initial state or measurements can lead to significant divergence in estimated states, making accurate estimation challenging. The system's inherent nonlinearity can cause non-Gaussian probability distributions for state variables, rendering conventional filters inadequate. To tackle these challenges, advanced techniques like GMM-type filters, particularly the EnGMF, excel. Its robustness in handling the chaotic, nonlinear, and non-Gaussian characteristics of the NRHO system makes it the preferred option for comparative analysis of different weight updates in GMM-type filters.

\subsection{Problem Setup}
A sensitivity study on the accuracy and consistency of each filter is performed for tracking a single target. The system dynamic equations are numerically integrated with an embedded Runge-Kutta 8(7) method \cite{ref:dormand}. Angle measurements are simulated using an Earth-based optical ground telescope mapped to the Barycenter of the system (see Fig.~\ref{fig:nrho-trajectory} for an illustration). In this simulation, tracking passes occur intermittently, with raw angle measurements available every 10 minutes. These measurements are grouped into 3 tracklets for each orbit, each lasting 2.5 hours. To simulate periods when the Moon is not visible, a break equivalent to a quarter of the orbit's period is inserted between each tracklet. This cycle repeats for 5 orbits ($\sim 30$ days): 2.5 hours of sensing at a 10-minute interval, followed by propagation for a quarter of the period, before resuming the sensing cycle for the next tracklet and/or orbit.

\begin{figure}[H]
    \centering
    \includegraphics[clip, trim=0cm 0cm 0cm 0cm, width=0.4\linewidth]{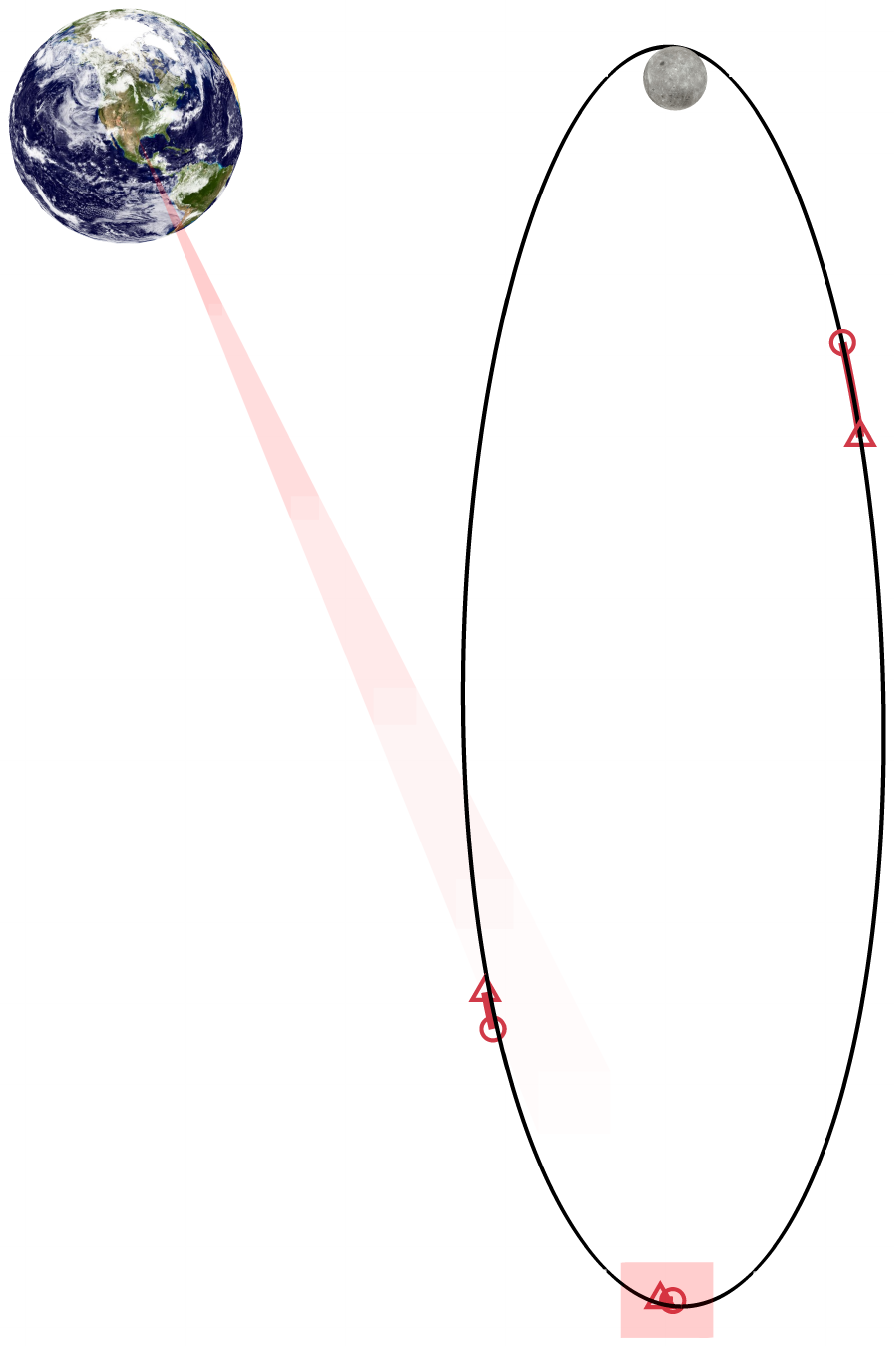}
    \caption{A single target in a Near Rectilinear Halo Orbit trajectory (black) with Earth-based optical ground sensing (\textit{objects not to scale}). There are 3 sensing windows per orbit lasting 2.5 hours each (red) denoted by a start (circle) and stop (triangle). The single target is tracked for 5 orbits $\sim$ 30 days.}
    \label{fig:nrho-trajectory}
\end{figure}

\subsubsection{Dynamics Model.}
This work models cislunar NRHO dynamics using the Circular Restricted Three Body Problem (CR3BP) for the Earth-Moon system with a 6 dimensional state-space represented by $x = [r_1, r_2, r_3, v_1, v_2, v_3]^{\mathrm{T}}$:
\begin{equation}
    \begin{aligned}
        \dot{r}_1 &= v_1, \ \ \ \ \ \
        \dot{r}_2 = v_2, \ \ \ \ \ \
        \dot{r}_3 = v_3, \\
        \dot{v}_1 &= r_1 + 2v_2 - \frac{(1-\mu)(r_1+\mu)}{r_{\Earth}^3} - \frac{\mu (r_1-1+\mu)}{r_{\leftmoon}^3}, \\
        \dot{v}_2 &= r_2 - 2v_1 - \frac{(1-\mu) r_2}{r_{\Earth}^3} - \frac{\mu r_2}{r_{\leftmoon}^3}, \\
        \dot{v}_3 &= -\frac{(1-\mu) r_3}{r_{\Earth}^3} - \frac{\mu r_3}{r_{\leftmoon}^3},
    \end{aligned}
\end{equation}
where $r_1,r_2,r_3$ and $v_1,v_2,v_3$ represent the scaled Cartesian positions and velocities of the satellite with-respect-to the Barycenter origin, $\mu$ is the scaled Moon geocentric gravitational constant, and $r_{\Earth}$ and $r_{\leftmoon}$ are the distances of the satellite with-respect-to the Earth and Moon in the Barycenter reference frame:
\begin{align}
    \mu\ &=\ \frac{\mu_{\leftmoon}}{\mu_{\Earth} + \mu_{\leftmoon}}, \\
    r_{\Earth}\ &=\ \sqrt{(r_1+\mu)^2\ +\ r_2^2\ +\ r_3^2}, \\
    r_{\leftmoon}\ &=\ \sqrt{(r_1-1+\mu)^2\ +\ r_2^2\ +\ r_3^2}.
\end{align}
In this work, $\mu_{\Earth} = G\cdot m_{\Earth}$ and $\mu_{\leftmoon} = G\cdot m_{\leftmoon}$. The gravitational constant is $G = 6.6743e-11$ [$\text{m}^3\text{s}^{-2}\text{kg}^{-1}$], the mass of the Earth is $m_{\Earth} = 5.972e24$ [kg], and the mass of the Moon is $m_{\leftmoon} = 7.342e22$ [kg]. The units for distance and time were non-dimensionalized by length units $\text{LU} = 384400e3$ [m] and time units $\text{TU} = \sqrt{\text{LU}^3 /\ (\mu_{\Earth} + \mu_{\leftmoon})}$ [s].

\subsubsection{Measurement Model.}
The raw measurement vector $\mathrm{y} = [\alpha, \ \delta]^{T}$ contains right ascension $\alpha$ and declination $\delta$ of the observed target as seen by an Earth-based optical ground telescope mapped to the Barycenter origin. The right ascension $\alpha$ and declination $\delta$ are
\begin{equation}
    \alpha\ =\ \tan^{-1} \left (\frac{r_2\ -\ r_{S_2}}{r_1\ -\ r_{S_1}} \right ),
\end{equation}
\begin{equation}
    \delta\ =\ \sin^{-1} \left (\frac{r_3\ -\ r_{S_3}}{ || [r_1, r_2, r_3]^{\mathrm{T}} - \boldsymbol{r_{S}} ||_2 } \right ),
\end{equation}
where $||\cdot||_2$ is the Euclidean 2-norm, and $\boldsymbol{r_{S}} = [r_{S_1}, r_{S_2}, r_{S_3}]^{\mathrm{T}}$ is the position of the telescope, which, in this work, is at the Barycenter of the system, defined to be the origin: $\boldsymbol{r_{S}} = [0, 0, 0]^{\mathrm{T}}$. The measurements are corrupted by additive zero-mean Gaussian white noise with standard deviation of 16.1 arc-seconds for both the right ascension and declination observations. Light travel time delay and measurement biases are not considered.

\subsubsection{Initial Conditions.} The target's truth and each EnGMF GMM component are initialized by the same distribution centered at the non-dimensional coordinates \cite{ref:spreen2021} and covariance \cite{ref:boonee2023}

\begin{align}
    \label{eqn:x0}
    x_0\ &=\ [1.0110350588,\ \   0,\ \  -0.1731500000,\ \\  
    &\phantom{=\ [} \ 0,\  \ -0.0780141199,\  \ 0]^{\mathrm{T}} \nonumber \\ 
    \label{eqn:P0}
    P_0\ &=\ \text{diag}([2.5e-5,\ \ 2.5e-5,\ \ 2.5e-5,\ \\
    &\phantom{=\ [} \ 1e-6,\ \ 1e-6,\ \ 1e-6]^2). \nonumber
\end{align}
The target is then propagated by 3 quarters of a period along it's trajectory after it's creation from \eqref{eqn:x0} and \eqref{eqn:P0}. This results in a new non-Gaussian initial condition due to the dynamics. The non-dimensional period of each target is roughly $1.3632096570$ from Reference~\cite{ref:spreen2021}.

\subsection{Estimation Criterion}
The accuracies of the methods are examined by their RMSE, which is computed from the true and estimated states at each measurement update for all Monte Carlo simulations. The RMSE is defined by \eqref{eqn:rmse} and is averaged over every time step and every Monte Carlo resulting in a single data point for each method. 

Additionally, the consistencies of the methods are examined using the scaled normalized estimation error squared (SNEES) which is defined as follows
\begin{equation}
    \text{SNEES}\ =\ \frac{1}{n_x} (x - \hat{x})^{T} \hat{P}^{-1} (x - \hat{x}),
\end{equation}
where similar to the RMSE, $n_x$ is the size of the state-space, $x$ is the truth, and $\hat{x}$ is the posterior state estimate. These also are averaged over every time step and every Monte Carlo resulting in a single SNEES data point for each method. 

Ideally the RMSE is as small as possible. Any value greater than zero suggests inaccuracies with respect to the truth. The size of the state-space $n_x = 6$ is used to scale the NEES value \cite{ref:barshalom2001} resulting in the SNEES. A SNEES value of 1 means good filter and tracking consistency. Anything less than 1 indicates that the filter is conservative and anything greater than 1 indicates that the filter is too confident.

\subsection{Results}
The Monte Carlo, time averaged RMSE and SNEES are computed for each EnGMF filtering implementation. As before, we varied the number of GMM components. The RMSE and SNEES results vs number of components are shown in Fig.~\ref{fig:nrho-rmse} and Fig.~\ref{fig:nrho-snees}, respectively. 

\begin{figure}[!ht]
    \centering
    \includegraphics[clip, trim=0cm 0cm 0cm 0cm, width=0.39\linewidth]{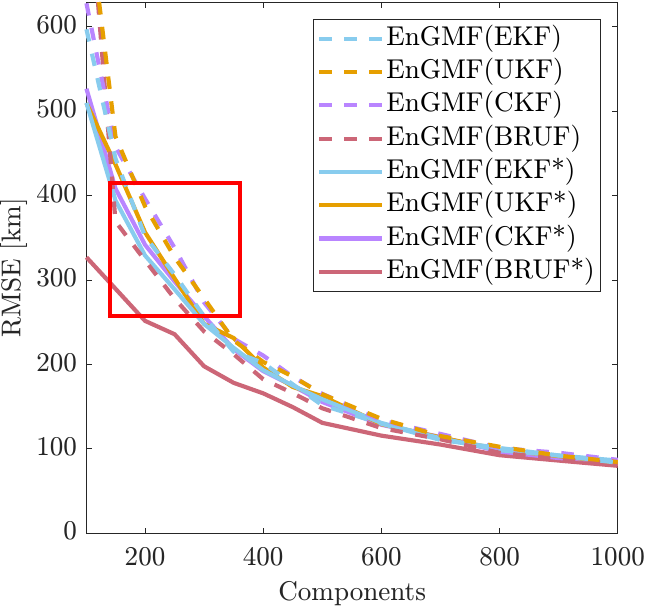}
    \includegraphics[clip, trim=0cm 0cm 0cm 0cm, width=0.4\linewidth]{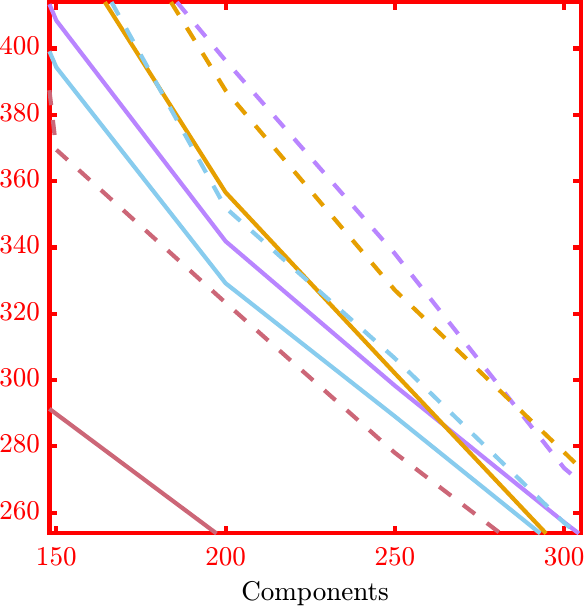}
    \caption{Comparing the impact on state estimate position accuracy with-respect-to the truth, represented by root mean square error (RMSE), by varying the number of EnGMF components in the cislunar NRHO example. Featuring the filters with traditional weights: EnGMF(EKF), EnGMF(UKF), EnGMF(CKF), and EnGMF(BRUF); and the filters with improved weights: EnGMF(EKF*), EnGMF(UKF*), EnGMF(CKF*), and EnGMF(BRUF*). Plots are averaged across 100 Monte Carlo simulations and 5 orbits $\sim 30$ days for each GMM component test case. A zoomed-in area is provided for clarity (red box). A good RMSE drives to zero.}
    \label{fig:nrho-rmse}
\end{figure}
\begin{figure}[!ht]
    \centering
    \includegraphics[clip, trim=0cm 0cm 0cm 0cm, width=0.39\linewidth]{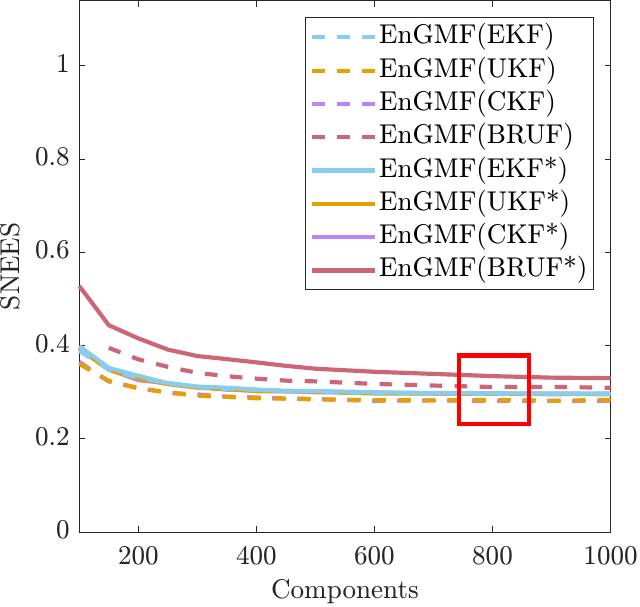}
    \includegraphics[clip, trim=0cm 0cm 0cm 0cm, width=0.4\linewidth]{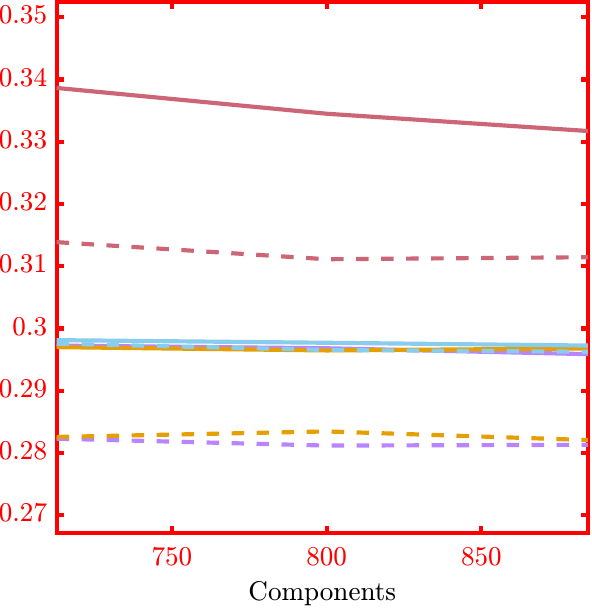}
    \caption{Comparing the impact on state estimate consistency with-respect-to the truth, represented by scaled normalized estimation error squared (SNEES), by varying the number of EnGMF components in the cislunar NRHO example. Featuring the filters with traditional weights: EnGMF(EKF), EnGMF(UKF), EnGMF(CKF), and EnGMF(BRUF); and the filters with improved weights: EnGMF(EKF*), EnGMF(UKF*), EnGMF(CKF*), and EnGMF(BRUF*). Plots are averaged across 100 Monte Carlo simulations and 5 orbits $\sim 30$ days for each GMM component test case. A zoomed-in area is provided for clarity (red box). A good SNEES sustains closer to, but below 1.}
    \label{fig:nrho-snees}
\end{figure}

In Fig.~\ref{fig:nrho-rmse}, with smaller number of components, there is an indication of improved performance when using the new weights over the traditional ones. As the number of components increases, the accuracies of the filters gradually converge towards each other no matter the weighting or update method; suggesting the difference between prior and posterior becomes smaller in the limit as the number of components goes to infinity. 

Although the accuracies seem to converge, the consistencies do not. In Fig.~\ref{fig:nrho-snees}, with larger number of components, there is a clear indication of improved consistency performance when using the new weights over the traditional ones. As the number of components increases, the consistencies of the filters stabilize; suggesting the new weights provide better posterior GMM reconstruction. This was indicated already when comparing the KLD in the single update examples in Fig.~\ref{fig:avocado-kld} and Fig.~\ref{fig:avocado-kld-sigma}.

\section{CONCLUSION}
\label{sec:summary}
This work introduces a novel method for computing weights in Gaussian mixture-type filters, demonstrating its equivalence to traditional methods for linear models and its superior performance for nonlinear cases. 
We showcase the wide-ranging versatility of the new GMM update weights by introducing formulations for established methods like the EKF and more contemporary approaches such as the BRUF.

Additionally, we propose a novel approach for conducting the GMM weight update using sigma points, accommodating both traditional and improved weight methods. Through empirical evaluations on the Avocado and cislunar NRHO examples, we showcase the practical advantages of our approach, surpassing traditional methods in terms of accuracy and consistency.

Moreover, our exploration of varying the number of Gaussian mixture components provides valuable insights into filter performance under different configurations. We observe improved accuracy and filter consistency, particularly in scenarios with fewer components, highlighting the practical applicability and resource effectiveness of our method.

Furthermore, future research could delve into the validity of \eqref{eqn:ibt}, which assumes that linearization about the posterior offers a better approximation of the truth compared to the prior. Investigating the performance changes when this assumption is not met could offer further insights into the method's robustness and potential areas for enhancement.




\section*{Acknowledgments}
This material is based on research sponsored by the Air Force Office of Scientific Research (AFOSR) under agreement number FA9550-23-1-0646, \textit{Create the Future Independent Research Effort (CFIRE)}.

\bibliographystyle{unsrtnat}
\bibliography{diss}

\end{document}